\documentclass{article}
\usepackage{arxiv}

\usepackage[utf8]{inputenc} 
\usepackage[T1]{fontenc}    
\usepackage{hyperref}       
\usepackage{url}            
\usepackage{booktabs}       
\usepackage{amsfonts}       
\usepackage{nicefrac}       
\usepackage{microtype}      
\usepackage{amsthm}
\usepackage{lipsum}
\usepackage{graphicx}
\usepackage{natbib}
\usepackage{amsmath}
\usepackage{tikz}
\usetikzlibrary{arrows.meta}
\usepackage{enumitem}
\usepackage{newtxtext}
\usepackage[subscriptcorrection]{newtxmath}

\graphicspath{ {./images/} }
\theoremstyle{definition}
\newtheorem{assumption}{Assumption}
\newtheorem{remark}{Remark}
\newtheorem{lemma}{Lemma}
\newtheorem{proposition}{Proposition}

\newcommand{\T}{{\mathrm{\scriptscriptstyle T}}}
\newcommand{\indep}{\perp\!\!\!\perp}

\newcommand{\pr}{\operatorname{pr}}

\title{Proximal causal inference through cross-proxy balancing}

\author{
Grace V. Ringlein \\
  Department of Biostatistics\\
  Johns Hopkins Bloomberg School of Public Health\\
  Maryland, USA \\
  \texttt{gringle1@jh.edu} \\
   \And
 Trang Quynh Nguyen \\
  Department of Mental Health\\
  Johns Hopkins Bloomberg School of Public Health\\
  Maryland, USA  \\
  \texttt{trang.nguyen@jhu.edu} \\
  \And
   Elizabeth A. Stuart \\
    Department of Biostatistics\\
  Johns Hopkins Bloomberg School of Public Health\\
  Maryland, USA \\
   \texttt{estuart@jhu.edu} \\
   \And
   Harsh Parikh \\
   Department of Biostatistics \\
   Yale University \\
   Connecticut, USA\\
   \texttt{harsh.parikh@yale.edu} \\
}

\begin{document}
\maketitle
\begin{abstract}
Proximal causal inference identifies causal effects in the presence of unmeasured confounding by drawing on two sets of proxy variables.  Identification typically utilizes bridge functions, defined as solutions to integral equations.
However, the mechanism by which fitted bridge functions correct for confounding bias remains opaque, offering little to interpret, inspect or stress-test. 
We show that the defining equation of a treatment bridge function is already a balance condition, with a cross-proxy form: the weights are functions of the treatment confounding proxies and covariates, and they balance the outcome confounding proxies and covariates (i.e., reweighting the distribution in a particular treatment arm to match the distribution across treatment arms). Several existing identification paths via a treatment bridge function can then be interpreted as providing conditions under which balance on the outcome proxies implies balance on the unobserved confounders, which we call balance propagation. Leveraging this framing, we show that estimation of the treatment bridge function is a type of balancing weight estimation. Finally, we show that an outcome-weighted estimator form can also be obtained for a large class of proximal estimators, including those that utilize an outcome bridge function. Using this framing, we provide conditions under which common estimators are numerically equivalent.
\end{abstract}


\section{Introduction}

Proximal causal inference (PCI) is a framework that allows identification and estimation of causal effects in the presence of unobserved confounders \citep{miao_identifying_2018,tchetgen_tchetgen_introduction_2024}. PCI methods typically leverage two types of auxiliary variables: treatment confounding proxies (assumed to be conditionally independent of the outcome given treatment assignment, covariates, and unobserved confounders) and outcome confounding proxies (assumed to be conditionally independent of the treatment assignment and treatment confounding proxies given covariates and unobserved confounders). 
There are two general paths to identification of causal effects: (i) via an
outcome bridge function \citep{miao_identifying_2018,tchetgen_tchetgen_introduction_2024,kallus_causal_2022} and (ii) via a treatment bridge function \citep{cui_semiparametric_2024,kallus_causal_2022}, which have parallels to standard g-formula and inverse probability weighting (IPW) identification approaches, respectively. These identification results motivate proximal g-computation estimators and proximal IPW estimators, which can be combined into augmented proximal AIPW estimators \citep{tchetgen_tchetgen_introduction_2024,cui_semiparametric_2024,ghassami_minimax_2021,kallus_causal_2022}.  

While some earlier work discusses the intuition behind these bridge functions \citep{miao_confounding_2024,ringlein_demystifying_2025}, the mechanisms underlying identification in PCI 
are still fairly opaque. Identification and estimation through a treatment bridge function in particular has received less attention than the outcome bridge function used in \citet{miao_identifying_2018}. 

In this work, we establish connections between the treatment bridge function and the covariate balancing literature \citep[see][ for a review]{ben-michael_balancing_2021} that provide new interpretations of existing identification results and estimators. We derive implied weights for a broad set of proximal estimators, including those built on an outcome bridge function or both types of bridge functions.  These forms are used to describe a class of estimators where the proximal AIPW estimator reduces to the proximal IPW estimator, up to a term linear in any remaining imbalance. In the special case where bridge functions are assumed to be linear within treatment arms and estimated in a particular way, all three estimators coincide.
\section{Preliminaries}

\subsection{Notation}
The full data consist of $i=1,\dots,n$ independent and identically distributed observations $\{\{Y_i(a): a\in\mathcal{A}\},A_i,U_i,X_i,W_i,Z_i\}$, where $A_i$ is treatment or exposure (with a finite set of treatment levels $a\in\mathcal{A}$), $Y_i(a)$ is the potential outcome under treatment $a$, $U_i$ denotes the unobserved confounders, 
treatment confounding proxies and outcome confounding proxies (formally defined by A.\ref{base}(d)--(e)) are $Z_i$ and  $W_i$, respectively, and the remaining covariates are $X_i$ with support $\mathcal{X}$. The set of square-integrable functions of a variable $V$ (e.g., $U$, $W$, or $Z$) given $a \in \mathcal{A}$ and $x \in \mathcal{X}$ is $\mathcal{L}_{ax}^2(V) = \{g : \int g(v)^2 \,dP_{V\mid A=a,X=x}< \infty\}$.  



\subsection{Base assumptions for proximal causal inference}
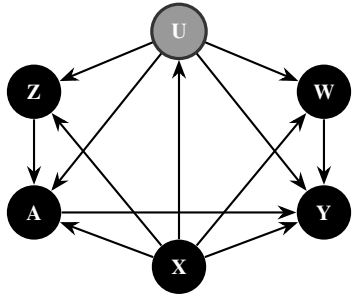
\begin{figure}
\centering
\begin{tikzpicture}[
    scale=0.80, transform shape,
    >={Stealth[length=2.5mm]},
    obs/.style={circle, fill=black, text=white, font=\bfseries,
                minimum size=0.9cm, inner sep=0pt},
    unobs/.style={circle, fill=gray!80, draw=black!80, line width=1.2pt,
                  text=white, font=\bfseries,
                  minimum size=0.9cm, inner sep=0pt},
    edge/.style={->, line width=0.8pt}
]
\node[unobs] (U) at (0, 3)    {U};
\node[obs]   (Z) at (-2.4, 2) {Z};
\node[obs]   (W) at (2.4, 2)  {W};
\node[obs]   (A) at (-2.4, 0) {A};
\node[obs]   (Y) at (2.4, 0)  {Y};
\node[obs]   (X) at (0, -0.9) {X};

\foreach \t in {Z, W, A, Y} \draw[edge] (U) -- (\t);
\foreach \t in {U, Z, W, A, Y} \draw[edge] (X) -- (\t);
\draw[edge] (Z) -- (A);
\draw[edge] (W) -- (Y);
\draw[edge] (A) -- (Y);
\end{tikzpicture}
\caption{Causal diagram with treatment $A$, outcome $Y$, observed
confounders $X$, unobserved confounders $U$, treatment confounding proxies $Z$ and outcome confounding proxies $W$, satisfying proximal conditional independence assumptions.}
\label{fig:dag}
\end{figure}
The following assumptions are common across proximal identification strategies \citep{miao_identifying_2018,tchetgen_tchetgen_introduction_2024,cui_semiparametric_2024,miao_confounding_2024,kallus_causal_2022}; we group them into a single assumption for succinctness: 

\begin{assumption}[Base Proximal Assumptions]\label{base}\hfill
\begin{enumerate}[label = \alph*.]
    \item Consistency: $Y_i = Y_i(A_i)$.
    \item Latent positivity: $\pr(A = a \mid U, X)\ge \epsilon$ almost surely, for some $\epsilon>0$ and all $a \in \mathcal{A}$.
    \item Latent unconfoundedness:  $Y(a) \indep A \mid U, X$ for $a \in \mathcal{A}$.
    \item Treatment confounding proxy conditional independence: $Z \indep Y \mid U, A, X.$
    \item Outcome confounding proxy conditional independence: $W \indep (A, Z) \mid U, X$.
\end{enumerate}
\end{assumption}

Figure \ref{fig:dag} is an example causal diagram where A.\ref{base}(c)--(e) are satisfied. Identifying $E\{Y(a)\}$ requires additional assumptions beyond A.\ref{base}. Typically, sufficient conditions are of two types: assumptions about the existence of solutions (called bridge functions) to particular integral equations, and completeness conditions. We explore identification via the less-studied treatment bridge functions and connections to balance in \S \ref{sec_bal}.

\section{The role of balance in identification with a treatment bridge function}\label{sec_bal}

\subsection{A treatment bridge function is a balancing weight by definition}\label{sec_qbal}
Following the notation of \citet{kallus_causal_2022}, $\mathbb{Q}_{obs}(a)$ is the set of ``observed-data'' treatment bridge functions $q_{obs}(Z,a,X)$ within treatment arm $a$, defined as solutions to \eqref{PCIq:eq}, and $\mathbb{Q}_{true}(a)$ is the set of ``true'' treatment bridge functions $q_{true}(Z,a,X)$ satisfying \eqref{PCIqu}.
\begin{align}
    \pr(A=a\mid W,X)^{-1}&=E\{q_{obs}(Z,a,X)\mid W,a,X\} \quad (a\in \mathcal{A}).\label{PCIq:eq}\\
     \pr(A=a\mid U,X)^{-1}&=E\{q_{true}(Z,a,X)\mid U,a,X\} \quad (a\in \mathcal{A}).\label{PCIqu}
\end{align}
There is not any guarantee that either set $\mathbb{Q}_{true}(a)$ or $\mathbb{Q}_{obs}(a)$ is non-empty. From here on, we abbreviate $\mathbb{Q}_{obs}(a)$ as $\mathbb{Q}_{obs}$ (and the same for $\mathbb{Q}_{true}$), leaving the particular arm $a$ implicit. 

Rearranging equations~\eqref{PCIq:eq} and \eqref{PCIqu} (see \S\ref{s:rearrange}) exposes the connection to balancing weights:
\begin{align}
E\{k(W,X)\}&=E\{I(A=a)\,q_{obs}(Z,a,X)\,k(W,X)\},\label{PCIqw_bal}\tag{\protect\ref{PCIq:eq}*}\\
E\{m(U,X)\}&=E\{I(A=a)\,q_{true}(Z,a,X)\,m(U,X)\},\label{PCIqu_bal}\tag{\protect\ref{PCIqu}*}
\end{align}
for any $k$ and $m$ such that expectations on both sides of the respective equations exist. Equations~\eqref{PCIqw_bal} and \eqref{PCIqu_bal} have the same general form as the standard balancing condition $$E\{I(A=a)\rho(a,X) k({X})\}=E\{k({X})\}$$ used in settings where unconfoundedness conditional on observed covariates, $Y(a) \indep A \mid X$, is assumed (combined with consistency and positivity conditional on $X$). 
In \eqref{PCIqw_bal}, $q_{obs}$ is a function of $Z,X$ for which the weighted mean of any function $k(W,X)$ within treatment arm $a$ (weighted by $q_{obs}$) equals the marginal (unweighted) mean of $k(W,X)$. We refer to \eqref{PCIqw_bal} (and equivalently, eq. \ref{PCIq:eq}) as a balancing condition that ``$q_{obs}$ balances $W,X$'', and to \eqref{PCIqu_bal} (and eq. \ref{PCIqu}) as the condition that ``$q_{true}$ balances $U,X$''.
A difference from the standard balancing condition is that  \eqref{PCIqw_bal} has what we refer to as cross-proxy structure: the weights $q_{obs}(Z,a,X)$ are a function of $Z,X$ but they balance $W,X$, while standard balancing weights $\rho(a,X)$ are functions of the same variables they balance ($X$). 

Under A.\ref{base}(e) ($W\indep (Z,A)\mid U,X$), any $q_{true}$ is also a $q_{obs}$, that is, $\mathbb{Q}_{true}\subseteq\mathbb{Q}_{obs}$ (see \S\ref{s:qtrueqobs}): thus $q_{true}$ also balances $W,X$. However, $\mathbb{Q}_{obs}$ and $\mathbb{Q}_{true}$ are not generally equivalent; without further assumptions there is no guarantee that $q_{obs}$ balances $U,X$ (i.e., solves \eqref{PCIqu}). In the next subsection, we discuss identification approaches that provide conditions under which any $q_{obs}\in \mathbb{Q}_{obs}$ does balance $U,X$; we refer to this as balance propagation (formal definition to follow).

\subsection{Identification under balance propagation}\label{sec_bp}

We consider identification under the assumption that at least one $q_{true}$ exists:
\begin{assumption}\label{PCIq}
    There is a solution to \eqref{PCIqu}: $\mathbb{Q}_{true}\not=\emptyset$.
\end{assumption}

\noindent If such a $q_{true}$ were known, under A.\ref{base}(a)--(d), $E\{Y(a)\}$ would be identified \citep[see \S \ref{s:qtrueid} or][]{cui_semiparametric_2024}:
\begin{align}\label{eq:q_true_id}
    E\{I(A=a)q_{true}(Z,a,X)Y\}=E\{Y(a)\},\quad (\forall q_{true}\in \mathbb{Q}_{true}).
\end{align}
As elements of $\mathbb{Q}_{true}$ are generally unknown, we rely on a $q_{obs}$ instead for identification. A.\ref{PCIq} implies the existence of at least one $q_{obs}$ (because $\mathbb{Q}_{true}\subseteq\mathbb{Q}_{obs}$). Identification proceeds by showing that \eqref{q_obs_id} (below) holds under additional assumptions: 
\begin{align}
    E\{I(A=a)q_{obs}(Z,a,X)Y\}&=E\{Y(a)\}, \quad (\forall q_{obs}\in \mathbb{Q}_{obs}).\label{q_obs_id}
\end{align}

Using \eqref{eq:q_true_id}, \eqref{q_obs_id} can be written in the equivalent form \eqref{q_id_condition}:
\begin{align}
    E[I(A=a)\{q_{obs}(Z,a,X)-q_{true}(Z,a,X)\}Y]&=0, \quad \{
    \forall (q_{obs},q_{true})\in \mathbb{Q}_{obs}\times\mathbb{Q}_{true}\}.\label{q_id_condition}
\end{align}
To see the role of balance propagation in identification, we can rewrite \eqref{q_id_condition} as (see \S\ref{s:ucond}):
\begin{align}\label{bal_prop_id}
E[\pr(A=a\mid U,X)\,E\{q_{obs}(Z,a,X)-q_{true}(Z,a,X)\mid U,a,X\}\,E(Y\mid U,a,X)]=0.
\end{align}
Thus under A.\ref{base} and A.\ref{PCIq}, identification in the form of \eqref{q_id_condition} holds if and only if \eqref{bal_prop_id} holds for every pair $(q_{obs},q_{true})\in\mathbb{Q}_{obs}\times\mathbb{Q}_{true}$. Clearly, one way for \eqref{bal_prop_id} to hold is if the middle term $E\{q_{obs}(Z,a,X)-q_{true}(Z,a,X)\mid U,a,X\}=0$ for all pairs $(q_{obs},q_{true})$. Formally: 
\begin{assumption}[Balance propagation]\label{a_balprop}
For $a\in\mathcal{A}$:
\begin{equation}\label{eq:bal_prop1}
    E\{q_{obs}(Z,a,X)\mid U,A=a,X\}=\pr(A=a\mid U,X)^{-1}\text{ a.s.} ,\quad (\forall q_{obs}\in \mathbb{Q}_{obs}).
\end{equation}
\end{assumption}

\noindent Using the form of $q_{true}$ \eqref{PCIqu} on the RHS and rearranging, \eqref{eq:bal_prop1} can be written as:
$$E\{q_{obs}(Z,a,X)-q_{true}(Z,a,X)\mid U,A=a,X\}=0 \text{ a.s.,} \quad \{\forall (q_{obs},q_{true})\in \mathbb{Q}_{obs}\times\mathbb{Q}_{true}\}.$$
Thus, under A.\ref{base}, A.\ref{PCIq}, and A.\ref{a_balprop}, $E\{Y(a)\}$ is identified by \eqref{q_obs_id}. We call A.\ref{a_balprop} balance propagation because, by \S\ref{sec_qbal}, it says that any $q_{obs}$ (which balances $W,X$ by definition) also balances $U,X$ ($\mathbb{Q}_{obs}\subseteq\mathbb{Q}_{true}$). As $\mathbb{Q}_{true}\subseteq\mathbb{Q}_{obs}$, A.\ref{a_balprop} is equivalent to $\mathbb{Q}_{obs}=\mathbb{Q}_{true}$. 

The two identification paths in \citet{cui_semiparametric_2024} use assumptions which can be viewed as sufficient conditions for A.\ref{a_balprop}, though neither their connection to A.\ref{a_balprop} nor the balancing interpretation appears in that work. 
The first of those two paths \citep[Theorem 2.2 in][]{cui_semiparametric_2024} uses a completeness condition (A.\ref{compUW}, below): a requirement on how richly $W$ varies with $U$ in arm $a$, for any strata of $X$. (For discrete variables it requires $W$ to have at least as many support points as $U$.)  
\begin{assumption}[Completeness of $W$ for $U$ given $A,X$]\label{compUW} For $a\in\mathcal{A}$ and  $P_{X}$-almost every $x\in \mathcal{X}$:
(a) for all 
 $g\in \mathcal{L}^2_{ax}(U)$: $E\{g(U) \mid W,A=a, X=x\} = 0 \text{ a.s.} \implies$ $g(U) = 0 \text{ a.s.}$, and (b) $E\{q_{obs}(Z,a,x)\mid U,a,x\}\in \mathcal{L}^2_{ax}(U)$ for all $q_{obs}\in \mathbb{Q}_{obs}.$
\end{assumption}
\noindent A.\ref{compUW}(a) can be applied to $E[E\{q_{obs}(Z,a,x)-q_{true}(Z,a,x)\mid U,a,x\}\mid W,a,x]$=0 (from $\mathbb{Q}_{true}\subseteq \mathbb{Q}_{obs}$ and A.\ref{base}) to obtain A.\ref{a_balprop} (see \S\ref{s:comp_implies_balprop}), provided that $E\{q_{obs}(Z,a,x)-q_{true}(Z,a,x)\mid U,a,x\}$ is square integrable: hence, why (b) is included. 
This square integrability condition is implicit in \citet{cui_semiparametric_2024} in their use of A.\ref{compUW}(a). We discuss this further in \S\ref{s:q_true_Lsqr}.

\begin{remark}
    Technically, this ``main'' identification path described in \citet{cui_semiparametric_2024} assumes that $\mathbb{Q}_{obs}\not=\emptyset$, rather than A.\ref{PCIqu} ($\mathbb{Q}_{true}\not=\emptyset$). Under A.\ref{compUW}, these are equivalent. 
\end{remark}


The second identification path \citep[Remark 5,][]{cui_semiparametric_2024} rests on A.\ref{base}, A.\ref{PCIqu}, and a condition (A.\ref{compZW}, below) that makes $q_{obs}$ unique (i.e., $\mathbb{Q}_{obs}=\{q_0\}$). Thus, by $\mathbb{Q}_{true}\subseteq \mathbb{Q}_{obs}$ we have $\mathbb{Q}_{true}=\mathbb{Q}_{obs}=\{q_0\}$. Therefore, $q_{obs}-q_{true}=q_{0}-q_{0}=0$ and A.\ref{a_balprop} holds. 

\begin{assumption}[Completeness of $W$ for $Z$ given $A,X$]\label{compZW}
For $a\in\mathcal{A}$ and  $P_{X}$-almost every $x\in \mathcal{X}$: (a) $\forall g\in \mathcal{L}_{ax}^2(Z)$: $E\{g(Z) \mid W,a,x\} = 0$ a.s. $\implies g(Z) = 0$ a.s., and (b) $\mathbb{Q}_{obs}\subseteq\mathcal{L}_{ax}^2(Z).$
\end{assumption}
\begin{remark}
    The insight that both paths of \citet{cui_semiparametric_2024} rely on balance propagation (A.\ref{a_balprop}) does not necessarily provide a condition that can be any more meaningfully discerned in practice than A.\ref{compUW} or A.\ref{compZW}. However, it does help to clarify the role of various completeness conditions used in \citet{cui_semiparametric_2024} and allows us to more easily compare those identification strategies to that of \citet{kallus_causal_2022} (which does not rely on balance propagation), discussed in the next section. 
\end{remark}

\subsection{Identification with a treatment bridge function without balance propagation}\label{sec_nobp}

Balance propagation is not the only route to identification with a treatment bridge function; however, the alternative in \citet{kallus_causal_2022} (their Theorem 1.2) also has a connection to balance, which we explore here. Starting from another equivalent form of \eqref{q_id_condition} (see \S \ref{s:id_z_condtion}): 
\begin{align}\label{bal_prop_idZ}
E\big\{I(A=a)\{q_{obs}(Z,a,X)-q_{true}(Z,a,X)\}E(Y\mid Z,a,X)\big\}=0,
\end{align}
identification along this route rests (still) on A.\ref{base} and A.\ref{PCIq}, now with the following assumption:
\begin{assumption}\label{PCIh}
    An outcome bridge function $h_{obs}$ exists, solving
\begin{equation}\label{PCIh:eq}
E\{Y\mid Z,a,X\}=E\{h_{obs}(W,a,X)\mid Z,a,X\} \quad (a\in\mathcal{A}).
\end{equation}
\end{assumption}
\noindent Under A.\ref{PCIh}, we may replace $E(Y\mid Z,a,X)$ in \eqref{bal_prop_idZ} by $E\{h_{obs}(W,a,X)\mid Z,a,X\}$.
Identification is then a direct consequence of the balance condition \eqref{PCIqw_bal}: if $h_{obs}(W,a,X)$ exists, any $q_{obs}$ balances it by definition (take $k(W,X)=h_{obs}(W,a,X)$ in \eqref{PCIqw_bal}), and so does $q_{true}$ (because $\mathbb{Q}_{true}\subseteq\mathbb{Q}_{obs}$); therefore \eqref{bal_prop_idZ} holds (see \S \ref{s:kmu} for details). 
Thus, while differing in exact assumptions used, the treatment bridge function identification paths of \citet{cui_semiparametric_2024} and \cite{kallus_causal_2022} fundamentally rely on the role of the treatment bridge functions as balancing weights. However, \cite{kallus_causal_2022} requires that one class of functions ($h_{obs}(W,a,X)$) be balanced by $q_{obs}$, while \citet{cui_semiparametric_2024} requires that all functions of $U,X$ be balanced by $q_{obs}$.



\section{On the implied weights of proximal IPW, g-computation, and AIPW estimators }
\subsection{The treatment bridge function is estimated as a balancing weight}\label{sec_pipw_est}

The proximal IPW estimator based on \eqref{q_obs_id}, denoted $\hat{\mu}_q(a)$, naturally takes the form of an outcome weighted estimator, with $\hat\omega_{q,i}(a)=n^{-1}I(A_i=a)\hat q_{obs}(Z_i,a,X_i)$:
\begin{align}\label{mu_q}
    \hat\mu_{q}(a) &=n^{-1}\sum_{i=1}^n I(A_i=a)\hat q_{obs}(Z_i,a,X_i)Y_i=\sum_{i=1}^n\hat\omega_{q,i}(a)Y_i.
\end{align}

Using matrices and vectors simplifies notation. Let $\hat\omega(a)$ denote a vector of outcome weights for units $i=1,\dots,n$, with transpose $\hat\omega(a)^{\T}$; let $Y=(Y_1,\dots,Y_n)^{\T}$; let $\hat q(a)$ be the $n$-vector with entries $\hat q_{obs}(Z_i,a,X_i)$; let $D(a)$ be the $n\times n$ diagonal matrix with $I(A_i=a)$ on the diagonal and zero elsewhere; and let $\mathbb{I}_n$ be the $n\times n$ identity matrix and $1_n$ the $n$-vector of ones. Then: $\hat \omega_q(a)^{\T}=n^{-1}\hat q(a)^{\T}D(a)$. Weights for units not in arm $a$ are 0; note the trailing $D(a)$.

What remains is to estimate $\hat q(a)$. Begin by considering a semiparametric setting where $q_{obs}$ is assumed to have a particular parametric form (e.g., $q(Z_i,a,X_i;\phi)=1+\exp(\phi_{a0} +\phi_{az}Z_i+\phi_{ax}X_i)$ for binary $A$).
Estimation is then based on the empirical analogue of \eqref{PCIqw_bal}:
\begin{equation}\label{eq:q-moments}
n^{-1}\sum_{i=1}^nI(A_i=a)\,q(Z_i,a,X_i)\,k_q^p(W_i,X_i)= n^{-1}\sum_{i=1}^n k_q^p(W_i,X_i) \quad (a\in\mathcal{A}),
\end{equation}
where $k_q^p(W_i,X_i)$ is  a vector of $p$ functions of $W_i,X_i$ to balance (e.g., $(1,W_i^\T,X_i^\T)^\T$). If $k_q^p(W_i,X_i)$ includes an intercept, \eqref{eq:q-moments} enforces that the weights sum to 1 (up to some error tolerance): $\hat \omega_q(a)^{\T}1_n=n^{-1}\hat q(a)^{\T}D(a)1_n\approx1$. As weights for units not in arm $a$ are exactly zero, this means the weights in arm $a$ alone sum to (approximately) 1. Weights $\hat \omega_{q,i}(a)$ are generally not constrained to be non-negative, though assuming a specific parametric form like $q(Z_i,a,X_i;\phi)=1+\exp(\phi_{a0} +\phi_{az}Z_i+\phi_{ax}X_i)$  may enforce this.

In the special case where a linear parametric form $q(Z_i,a,X_i)=\phi_{a0} + \phi_{az}^\T Z_i + \phi_{ax}^\T X_i$ is assumed, and $k_q^p(W_i,X_i)=(1,W_i^\T,X_i^\T)^\T$ is balanced, and an invertibility condition requiring dim($W_i$)=dim($Z_i$) is satisfied, \eqref{eq:q-moments} can be solved exactly; we present this example in \S\ref{sec_dr_q_equiv}. 

Generally, \eqref{eq:q-moments} may be  solved as an optimization problem, 
minimizing a loss function related to the imbalance ($r_n=n^{-1}(\hat q(a)^{\T}D(a)K_{wx}-1_n^{\T}K_{wx})$, where $K_{wx}$ is a matrix with rows $k_q^p(W_i,X_i)^{\T}$), potentially subject to a penalty on the complexity of the weights, as similarly done in estimating balancing weights in the standard setting \citep{ben-michael_balancing_2021}. 
For example, \citet{cui_semiparametric_2024} suggest estimation via generalized method of moments, which finds the solution that minimizes a quadratic form in the imbalance $r_n$ as the loss function. Balancing interpretations are not limited to semiparametric methods; 
see \S \ref{s:lin_h} for discussion.


\subsection{On balance checks with the estimated treatment bridge function} In typical propensity score estimation approaches that do not explicitly target balance (e.g., estimating $\pr(A=1\mid X)$ via logistic regression), balance is a consequence of well-estimated weights, and post hoc balance checks with the estimated weights are recommended. When balance on a specific set of factors is targeted in estimating the weights, balance holds by construction for these factors. While one can still conduct checks \citep[e.g., on higher order moments not included in the set of terms that are balanced, as in][]{chattopadhyay_implied_2021}, moments believed relevant to confounding should be included as terms to balance \citep{hainmueller_entropy_2012}. 

Similarly, $\hat q(a)$ estimation targets balance on $W$ and $X$ so balance is satisfied by construction; still, we recommend conducting checks (i.e., computing $r_n$) to verify that the estimate of $\hat q(a)$ has converged properly. We recommend checking balance specifically on the outcome bridge function $\hat h_{obs}(W,a,X)$ when it has also been estimated, as the ability of $q_{obs}(Z,a,X)$ to balance $h_{obs}(W,a,X)$ is fundamental in identification (see \S\ref{sec_nobp}). That said, the utility of this check still relies on  $\hat h_{obs}(W,a,X)$ being well estimated (e.g., parametric model is correctly specified). All of the identification methods discussed still rely on fundamentally untestable assumptions. Even if the estimated $\hat q_{obs}(a)$ does a good job of balancing functions of $W,X$, whether balance propagates to balance on functions of $U,X$ (A.\ref{a_balprop}) cannot be checked. Neither can sufficient conditions for balance propagation (A.\ref{compUW}, A.\ref{compZW}), in general.  Finally, it should be clarified that $\hat{q}(a)$ aims to balance $W,X$, not $Z$, so any balance checks done should use functions of the outcome confounding, but not treatment confounding proxies. 

\subsection{Outcome-weighted estimator forms}\label{sec_ow}

So far, we have focused primarily on identification with a treatment bridge function; however, an observed-data outcome bridge function $h_{obs}(W,a,X)$ defined as solving \eqref{PCIh:eq} provides an alternate identification route. In short, under different sets of identifying conditions (involving $h_{true}(W,a,X)$ solving $E\{Y\mid U,a,X\}=E\{h_{true}(W,a,X)\mid U,a,X\}$),  the potential outcome mean can be identified by the proximal g-formula $E\{Y(a)\}=E\{h_{obs}(W,a,X)]$ \citep{miao_identifying_2018,tchetgen_tchetgen_introduction_2024}. See \S\ref{s:hmirror} for further discussion. This motivates proximal g-computation estimators of the form $\hat\mu_{h}(a) =n^{-1}\sum_{i=1}^n \hat h_{obs}(W_i,a,X_i)$ \citep{tchetgen_tchetgen_introduction_2024} and proximal AIPW estimators combining both functions \citep{cui_semiparametric_2024}: $\hat \mu_{qh}(a)=n^{-1}\sum_{i=1}^n\{I(A=a)\hat q_{obs}(Z,a,X)[Y_i-\hat h_{obs}(W_i,a,X_i)]+\hat h_{obs}(W_i,a,X_i)\}$.

Many estimators of $E\{Y(a)\}=\mu(a)$ can be written as outcome-weighted estimators: $\hat\mu(a)=\sum_{i=1}^n\hat\omega_i(a)Y_i$, where $\hat\omega_i(a)$ is the outcome weight for unit $i$ \citep{knaus_treatment_2024}, called the ``implied'' outcome weight when the estimator does not natively take an outcome weighted form \citep{chattopadhyay_implied_2021}. The proximal IPW estimator $\hat{\mu}_q(a)$ \eqref{mu_q} is naturally in this form with $\hat \omega_q(a)^{\T}=n^{-1}\hat q(a)^{\T}D(a)$ (\S\ref{sec_pipw_est}). 
The proximal g-computation estimator $\hat \mu_h(a)$  and proximal AIPW estimator $\hat \mu_{qh}(a)$ are also outcome-weighted estimators when $\hat h_{obs}(W_i,a,X_i)$ takes the form of a weighted average of outcomes across all units. That is, if $\hat h(W_i,a,X_i)=\hat s_i(a)^{\T}Y$, then collecting the $\hat s_i(a)^{\T}$ as the rows of an $n\times n$ outcome-smoother matrix $\hat S(a)$: 
\begin{align}
    \hat\mu_{h}(a) &=n^{-1}1_n^{\T}\hat S(a)Y=\hat\omega_h(a)^{\T}Y\label{mu_h}\\
\hat\mu_{qh}(a) 
    &=n^{-1}\{\hat q(a)^{\T}D(a)(\mathbb{I}_n - \hat S(a)) + 1_n^{\T}\hat S(a)\}Y=\hat\omega_{qh}(a)^{\T}Y.\label{mu_dr}
\end{align}
Thus, $\hat\omega_{qh}(a)^\T=\hat \omega_q(a)^{\T}- \hat \omega_q(a)^{\T}\hat S(a) + \hat \omega_h(a)^{\T}$.
Whether $\hat h(a)$ (abbreviating the vector of fitted values $\hat h_{obs}(W_i,a,X_i)$) has the form $\hat h(a)=\hat S(a)Y$ depends on how it is estimated. We show that several common methods satisfy this form: \S\ref{sec_dr_q_equiv} addresses when $h(W,a,X;\eta)$ is assumed to be linear in each arm and \S\ref{s:lin_h} provides additional examples. 

Assuming a parametric form of $h(W,a,X;\eta)$, the outcome bridge $\hat h(a)$ can be estimated from the empirical version of \eqref{PCIh:eq}. Letting $K_{zx}$ be a matrix with rows $k_h^p(Z_i,X_i)^\T$:
 \begin{equation}\label{eq:h-moments}
\{Y-h(a)\}^{\T}D(a)K_{zx}=0
\end{equation}


Equation \ref{eq:h-moments} can in some cases be solved directly (see \S\ref{sec_dr_q_equiv}), or approximated via generalized method of moments \citep{cui_semiparametric_2024,miao_confounding_2024}. Nonparametric estimation approaches can also be used and may satisfy the linear smoother form (see \S\ref{s:lin_h}).  
 
Thus, while $\hat q(a)$ is fitted so that functions of $(W,X)$ are balanced within  arm $a$ to the marginal mean, the outcome bridge function
$\hat h(a)$ is fitted by making the outcome residual orthogonal to functions of $(Z,X)$ within arm $a$. The proximal AIPW estimator combines the two. 


\begin{remark}\label{remark_ESS}
    
The effective sample size ($\text{ESS}=\{(\sum_{i=1}^n\hat\omega_i(a))^2\}\{\sum_{i=1}^n\hat\omega_i(a)^2\}^{-1}$, derived under assumptions such as homoscedasticity of outcomes) 
 can be used to reflect the loss of precision when weights are largely concentrated on a small number of units \citep{kish_survey_1965}. The weighting forms derived above suggest that ESS could similarly be used to quantify (the inverse of) variance inflation for proximal estimators; exploring this and other diagnostics from standard weighting methods for use in proximal causal inference is a direction of future work.
\end{remark}
\subsection{Equivalence of semiparametric estimators under particular conditions}\label{sec_dr_q_equiv}

In the standard setting, the augmented weighting estimator reduces to the weighting estimator when the outcome model is linear in the balanced functions \citep{sloczynski_covariate_2025}. 
We provide conditions for the analogous result for the proximal setting: when the estimated outcome bridge function lies in the span of the functions that $\hat q(a)$ is estimated to balance, the proximal AIPW estimator reduces to the proximal IPW estimator, up to a term linear in any imbalance. Formally:
\begin{proposition}\label{prop:collapse}
    Let $\hat q(a)$ be estimated to approximately solve the balance condition \eqref{eq:q-moments} with $K_{wx}$. If $\hat h(a)$ lies in the span of $K_{wx}$, such that $\hat h(a)=K_{wx}\hat\eta_a$, the corresponding proximal AIPW estimator $\hat\mu_{qh}(a)$ can be written as: $\hat\mu_{qh}(a) = \hat\mu_{q}(a) -r_n\hat\eta_a$, where $r_n=n^{-1}(\hat q(a)^{\T}D(a)K_{wx}-1_n^{\T}K_{wx})$.  If  \eqref{eq:q-moments} can be solved exactly (i.e., $r_n=0$), the estimators are numerically equivalent. 
\end{proposition}

\begin{proof}
    Plugging $\hat h(a)=K_{wx}\hat\eta_a$ into the form of the proximal AIPW estimator \eqref{mu_dr} gives $\hat\mu_{qh}(a) = \hat\mu_{q}(a) - n^{-1}\hat q(a)^{\T}D(a)K_{wx}\hat\eta_a + \hat\mu_h(a).$ Then, plug in the rearranged definition of the imbalance $n^{-1}\hat q(a)^{\T}D(a)K_{wx}=r_n+n^{-1}1_n^{\T}K_{wx}$, to get $\hat\mu_{q}(a) - n^{-1}1_n^{\T}K_{wx}\hat\eta_a -r_n\hat\eta_a + \hat\mu_h(a)$. Simplifying via $\hat\mu_h(a)=n^{-1}1_n^{\T}K_{wx}\hat\eta_a$  yields the proposed form.
\end{proof}

An example is when a linear form for $h(W,a,X)=\eta_{a0}+\eta_{aw}^\T W+\eta_{ax}^\T X$ is assumed (or similarly, linear across arms $h(W,A,X)=\eta_{0}+\eta_{a}A+\eta_{w}^\T W+\eta_{x}^\T X$; see \S\ref{s:lin_h}) and $q$ is estimated via \eqref{eq:q-moments} with $k_q^p(W_i,X_i)=(1, W_i^\T,X_i^\T)^{\T}$ such that the means of $W$ and $X$ are balanced (regardless of the parametric form of $q(Z,a,X)$ that is assumed). A further special case is when $q(Z,a,X;\phi)$ is also assumed to be linear and \eqref{eq:q-moments} can be solved exactly. Then, the weights $\hat\omega_{q}(a)$ are the same as the implied weights of the other two estimators $\hat\omega_{h}(a)$ and $\hat\omega_{qh}(a)$ under the assumption that $h(W,a,X)$ is linear in each arm (estimated via \eqref{eq:h-moments} with $k_h^p(Z_i,X_i)=(1,Z_i^{\T},X_i^{\T})^{\T}$). Formally: \begin{proposition}\label{prop:linear_pipw}
Let $\mathbb{K}_{zx}$ be the matrix with rows $(1,Z_i^\T,X_i^\T)$ and $\mathbb{K}_{wx}$ the matrix with rows $(1,W_i^\T,X_i^\T)$. Suppose $q(a)=\mathbb{K}_{zx}\phi_a$, \eqref{eq:q-moments} is solved with $\mathbb{K}_{wx}$ and $\mathbb{K}_{zx}^{\T}D(a)\mathbb{K}_{wx}$ is invertible (requiring dim($W$)=dim($Z$)). Additionally, suppose $h(a) = \mathbb{K}_{wx}\eta_a$ and \eqref{eq:h-moments} is solved with $\mathbb{K}_{zx}$. The corresponding proximal IPW, g-computation, and AIPW estimators of $E\{Y(a)\}$ (derived in \S \ref{s:linear_weights}) can all be written in the outcome weighted form $\hat\mu_{\text{lin}}(a)=\hat\omega_{\text{lin}}(a)^{\T} Y$ with the same implied weights 
$\hat\omega_{\text{lin}}(a)^{\T}=n^{-1}1_n^{\T}\mathbb{K}_{wx}(\mathbb{K}_{zx}^{\T}D(a)\mathbb{K}_{wx})^{-1}\mathbb{K}_{zx}^{\T}D(a).$
\end{proposition}

This implies that the estimator $\hat\mu_{\text{lin}}(a)$ has a layer of robustness not previously noted:  even if $h(a) = \mathbb{K}_{wx}\eta_a$ is misspecified, $\hat\mu_{\text{lin}}(a)$ can still provide a consistent estimate of $E\{Y(a)\}$ so long as if $q(a)=\mathbb{K}_{zx}\phi_a$, and vice versa (given identification assumptions). Proof of Proposition \ref{prop:linear_pipw} is in \S\ref{s:linear_weights}. A caveat for Proposition \ref{prop:linear_pipw} is that arm-wise linearity conflicts with $\sum_{a\in\mathcal{A}}\pr(A=a\mid W,X)=1$; this is a relevant if contrasts such as $E\{Y(1)\}-E\{Y(0)\}$ are of primary interest.

\section{Discussion}
The link we establish between
the proximal causal inference and balancing weight literatures provides a new interpretation and means of connection across existing identification results. Utilizing outcome weighted forms of common estimators, several equivalence results (\S\ref{sec_dr_q_equiv}) clarify settings in which estimators based on one or both bridge functions can be expected to be similar, and in some cases numerically equivalent. 

\bibliography{references}      
\clearpage
\begin{appendix}

\renewcommand{\thesection}{S\arabic{section}}
\renewcommand{\theequation}{s\arabic{equation}}
\renewcommand{\theproposition}{S\arabic{proposition}}
\renewcommand{\theremark}{S\arabic{remark}}
\renewcommand{\theassumption}{S\arabic{assumption}}
\renewcommand{\thelemma}{S\arabic{lemma}}

\vspace{\baselineskip}
\LARGE \textbf{Supplementary material for `Proximal causal inference through cross-proxy balancing'}

\normalsize \noindent Equation and result numbers of the form (1), Proposition 1, Assumption 1
refer to the main text. Sections here are numbered S0, S1, S2, \dots and are cited from
the main text as ``Supplementary Material, \S S$n$''.

\setcounter{section}{-1}
\setcounter{equation}{0}
\setcounter{proposition}{0}
\setcounter{lemma}{0}
\setcounter{remark}{0}
\setcounter{assumption}{0}
\section{A conditioning identity used repeatedly}\label{s:cond}
The following identity is applied repeatedly (e.g., with $V=(W,X)$ and $V=(U,X)$), to move the treatment
indicator in and out of conditional expectations: for any integrable $f$, any $a\in\mathcal{A}$, and any conditioning variable $V$, where $\pr(A=a\mid V)>0$:
$E\{I(A=a)f\mid V\}=\pr(A=a\mid V)E(f\mid A=a,V).$

\section{{Equivalence of the bridge-function definitions and the balance conditions}}\label{s:rearrange}

Here we show that the definition of the treatment bridge function in \eqref{PCIq:eq} ($\pr(A=a\mid W,X)^{-1}=E\{q_{obs}(Z,a,X)\mid W,a,X\}$) is equivalent to the balance condition \eqref{PCIqw_bal} ($E\{k(W,X)\}=E\{I(A=a)\,q_{obs}(Z,a,X)\,k(W,X)\}$) for all $k(W,X)$ such that both expectations exist.

Equation~\eqref{PCIqw_bal}  is obtained by multiplying both sides of
\eqref{PCIq:eq} by $\pr(A=a\mid W,X)$ and using the identity in \S\ref{s:cond}:
\begin{align*}
    1&=\pr(A=a\mid W,X)E\{q_{obs}(Z,a,X)\mid W,a,X\}=E\{I(A=a)q_{obs}(Z,a,X)\mid W,X\}.
\end{align*}
Subtracting one from both sides gives $E\{I(A=a)q_{obs}(Z,a,X)-1\mid W,X\}=0$. Multiplying by any $k(W,X)$ with $E|I(A=a)q_{obs}k|+E|k|<\infty$ and taking expectations
over $(W,X)$,
\begin{align*}
    E\big[\{I(A=a)q_{obs}(Z,a,X)-1\}k(W,X)\big]
    &=E\big[k(W,X)\,E\{I(A=a)q_{obs}(Z,a,X)-1\mid W,X\}\big]\\
    &=0,
\end{align*}
which rearranges to \eqref{PCIqw_bal}. Conversely, if \eqref{PCIqw_bal} holds for
all such $k$ then the conditional expectation vanishes almost surely, giving
\eqref{PCIq:eq}. The same steps show the equivalence of \eqref{PCIqu} and
\eqref{PCIqu_bal}, with $q_{true}$ in place of $q_{obs}$ and $m(U,X)$ in place of
$k(W,X)$.

\section{Any $q_{true}$ is also a $q_{obs}$}\label{s:qtrueqobs}

As shown in \citet{kallus_causal_2022}, $\mathbb{Q}_{true} \subseteq \mathbb{Q}_{obs}$. For $a\in \mathcal{A}$ and any $q_{true}\in \mathbb{Q}_{true}$:
\begin{align*}
    E\{&q_{true}(Z,a,X)\mid W,a,X\}\\
    &=E[E\{q_{true}(Z,a,X)\mid U,W,a,X\}\mid W,a,X]\tag{iterated exp.}\\
    &=E[E\{q_{true}(Z,a,X)\mid U,a,X\}\mid W,a,X]\tag{by A.\ref{base}(e): $W\indep Z\mid U,A,X$ }\\
    &=E\{\pr(A=a\mid U,X)^{-1}\mid W,a,X\}\tag{def. of $q_{true}$}\\
    &=\int\pr(A=a\mid u,X)^{-1}\,dP_{U\mid W,A=a,X}(u).
\end{align*}
The conditional law of $U$ given $(W,A=a,X)$ satisfies the change of measure
$$dP_{U\mid W,A=a,X}=\frac{\pr(A=a\mid U,W,X)}{\pr(A=a\mid W,X)}\,dP_{U\mid W,X}
=\frac{\pr(A=a\mid U,X)}{\pr(A=a\mid W,X)}\,dP_{U\mid W,X},$$
the second equality by $W\indep A\mid U,X$ (A.\ref{base}(e)). Substituting,
\begin{align*}
    \int\pr(A=a\mid u,X)^{-1}\,dP_{U\mid W,A=a,X}(u)
    &=\int\frac{dP_{U\mid W,X}(u)}{\pr(A=a\mid W,X)}
    =\pr(A=a\mid W,X)^{-1},
\end{align*}
so any $q_{true}\in\mathbb{Q}_{true}$ solves the observed-data equation \eqref{PCIq:eq}; therefore $\mathbb{Q}_{true} \subseteq \mathbb{Q}_{obs}$. Note that dividing by $\pr(A=a\mid W,X)$ in the change of measure requires $\pr(A=a\mid W,X)>0$, which holds under A.\ref{base}(b) ($\pr(A=a\mid U,X)\geq \epsilon >0$) and A.\ref{base}(e) ($W\indep A \mid U,X$): 

\begin{align*}
    \pr(A=a\mid W,X) &= E\{\pr(A=a\mid U,W,X)\mid W,X\}\\
    &=E\{\pr(A=a\mid U,X)\mid W,X\}\tag{A.\ref{base}(e)}\\
    &\geq E\{\epsilon \mid W,X\} = \epsilon > 0.\tag{A.\ref{base}(b)}
\end{align*}

\section{Expressing $E\{Y(a)\}$ in terms of $q_{true}$ }\label{s:qtrueid}

Assume $\mathbb{Q}_{true}\not=\emptyset$, (A.\ref{PCIq}). For any $q_{true}\in \mathbb{Q}_{true}$:
\begin{align*}
E\{&I(A=a)q_{true}(Z,a,X)Y\}\\
&=E[E\{I(A=a)q_{true}(Z,a,X)Y\mid U,X\}]\tag{by iterated exp.}\\
&=E[\pr(A=a\mid U,X) E\{q_{true}(Z,a,X)Y\mid U,A=a,X\}]\tag{by \S\ref{s:cond}}\\
&=E[\pr(A=a\mid U,X) E\{q_{true}(Z,a,X)\mid U,a,X\}E\{Y\mid U,a,X\}] \tag{by A.\ref{base}(d): $Z \indep Y \mid U,A,X$}\\
&=E[E\{Y\mid U,a,X\}] \tag{by def. of $q_{true}$}\\
&=E[E\{Y(a)\mid U,a,X\}] \tag{by A.\ref{base}(a): $Y(A_i)=Y_i$}\\
&=E[E\{Y(a)\mid U,X\}] \tag{by A.\ref{base}(c): $Y(a)\indep A \mid U,X$}\\
&=E\{Y(a)\}. \tag{rev. iterated exp.}
\end{align*}
This shows that \eqref{eq:q_true_id} holds under A.\ref{base}(a)-(d). This is similarly shown in \citet{cui_semiparametric_2024} and \citet{kallus_causal_2022}. 

\section{The identification condition \eqref{q_id_condition} in a $U$-conditional form \eqref{bal_prop_id}}\label{s:ucond}

Assuming $\mathbb{Q}_{true}\not=\emptyset$ (A.\ref{PCIq}), which implies $\mathbb{Q}_{obs}\not=\emptyset$ by \S\ref{s:qtrueqobs},  consider any pair $(q_{obs},q_{true})\in (\mathbb{Q}_{obs}\times\mathbb{Q}_{true}):$
\begin{align*}
E[&I(A=a)\{q_{obs}(Z,a,X)-q_{true}(Z,a,X)\}Y]\tag{LHS of \ref{q_id_condition}}\\
&=E[E[I(A=a)\{q_{obs}(Z,a,X)-q_{true}(Z,a,X)\}Y\mid U,X]]\tag{by iterated exp.}\\
&=E[\pr(A=a\mid U,X) E[\{q_{obs}(Z,a,X)-q_{true}(Z,a,X)\}Y\mid U,A=a,X]]\tag{by \S\ref{s:cond}}\\
&=E[\pr(A=a\mid U,X) E\{q_{obs}(Z,a,X)-q_{true}(Z,a,X)\mid U,a,X\}E\{Y\mid U,a,X\}]. \tag{by A.\ref{base}(d): $Z \indep Y \mid U,A,X$}
\end{align*}

The final line is the LHS of \eqref{bal_prop_id}. Thus, \eqref{bal_prop_id} holding for all $(q_{obs},q_{true})\in (\mathbb{Q}_{obs}\times\mathbb{Q}_{true})$ is equivalent to \eqref{q_id_condition} holding for all pairs $(q_{obs},q_{true})$. 

\section{Completeness of $W$ for $U$ implies balance propagation}\label{s:comp_implies_balprop}
Consider any $q_{obs}$ and any $q_{true}$, and let $\delta^U_{q,ax}(U)=E\{q_{obs}(Z,a,x)-q_{true}(Z,a,x)\mid U,a,x\}$ be evaluated at fixed $(a,x)$. A.\ref{compUW}(b) assumes that $E\{q_{obs}(Z,a,x)\mid U,a,x\}\in \mathcal{L}^2_{a,x}(U)$ for all $q_{obs}$; this and A.\ref{base}(b) imply $\delta^U_{q,ax}(U)=E\{q_{obs}(Z,a,x)-q_{true}(Z,a,x)\mid U,a,x\}$ is square integrable (see \S \ref{s:q_true_Lsqr}).  Then: 
\begin{align*}
E\{q_{obs}(Z,a,x)-q_{true}(Z,a,x)\mid W,a,x\}&=0\tag{$q_{true}\in \mathbb{Q}_{obs}$}\\
   E[E\{q_{obs}(Z,a,x)-q_{true}(Z,a,x)\mid W,U,a,x\}\mid W,a,x]&=0\tag{iterated expectation}\\
      E[E\{q_{obs}(Z,a,x)-q_{true}(Z,a,x)\mid U,a,x\}\mid W,a,x]&=0\tag{by $W\indep Z \mid U,a,x$}\\
      E\{\delta^U_{q,ax}(U)\mid W,a,x\}&=0\tag{definition of $\delta^U_{q,ax}$}\\
     \delta^U_{q,ax}(U)&=0 \text{ a.s.}\tag{by completeness A.\ref{compUW}}
\end{align*}

\section{On square integrability of treatment bridge functions}\label{s:q_true_Lsqr}
Square integrability of $\delta^U_{q,ax}(U)=E\{q_{obs}(Z,a,x)-q_{true}(Z,a,x)\mid U,a,x\}$ (or equivalently, $\delta^U_{q,ax}(U)=E\{q_{obs}(Z,a,x)\mid U,a,x\}-\pr(A=a\mid U,x)^{-1}$), is implicitly used in the identification in \citet{cui_semiparametric_2024}; we make it explicit here. Here we show that A.\ref{compUW}(b) (which assumed that $E\{q_{obs}(Z,a,x)\mid U,a,x\}\in \mathcal{L}^2_{a,x}(U)$) and latent positivity (A.\ref{base}(b)) are sufficient for $\delta^U_{q,ax}(U)=E\{q_{obs}(Z,a,x)-q_{true}(Z,a,x)\mid U,a,x\}$ to be square integrable; thus that A.\ref{compUW}(a) can be applied to $\delta^U_{q,ax}(U)=E\{q_{obs}(Z,a,x)-q_{true}(Z,a,x)\mid U,a,x\}$. 

The second moment of  $\delta^U_{q,ax}(U)$ can be bounded by the following: 
\begin{align*}
    E[\delta^U_{q,ax}(U)^2\mid a,x]&=E[E\{q_{obs}(Z,a,x)-q_{true}(Z,a,x)\mid U,a,x\}^2\mid a,x ]\\
   & \leq 2E[E\{q_{true}(Z,a,x)\mid U,a,x\}^2\mid a,x ]+2E[E\{q_{obs}(Z,a,x)\mid U,a,x\}^2\mid a,x ].
\end{align*}

The boundedness of the first term, $E[E\{q_{true}(Z,a,x)\mid U,a,x\}^2\mid a,x] < \infty$ can be derived from latent positivity (A.\ref{base}(b): $\pr(A=a\mid U,X)\geq \epsilon>0$ almost surely):
\begin{align*}
    E\{q_{true}(Z,a,x)\mid U,a,x\}^2&=\pr(A=a\mid U,x)^{-2}
    \leq \epsilon^{-2} <\infty\text{ a.s}.
\end{align*}
Boundedness of the second term, $E[E\{q_{obs}(Z,a,x)\mid U,a,x\}^2\mid a,x] < \infty$ is assumed by A.\ref{compUW}(b). Thus, under the specified conditions, $E[\delta^U_{q,ax}(U)^2\mid a,x]<\infty$. 

Square integrability of $E\{q_{obs}(Z,a,x)\mid U,a,x\}$ can alternatively be shown under the assumption $q_{obs}(Z,a,x)\in \mathcal{L}^2_{ax}(Z)$, $\forall q_{obs}\in Q_{obs}$, rather than by assumption A.\ref{compUW}(b). Sufficient conditions for the existence of a $q_{obs}\in \mathcal{L}^2_{ax}(Z)$ are provided in \citet{cui_semiparametric_2024} (in their Supplementary Material, part B): $E[g(W)\mid Z,a,x]=0$ a.s. implies $g(W)=0$ a.s. for $a\in \mathcal{A}$ and all $x\in \mathcal{X}$ and regularity conditions (such that Picard's theorem, based on singular value decomposition, can be applied). Under these conditions $E\{q_{obs}(Z,a,x)\mid U,a,x\}$ is also square integrable; this can be shown using iterated expectation and Jensen's inequality, 
\begin{align*}
    \infty > E\{q_{obs}(Z,a,x)^2\mid a,x\}&=E[E\{q_{obs}(Z,a,x)^2\mid U,a,x\}\mid a,x]\\
    &\geq E[E\{q_{obs}(Z,a,x)\mid U,a,x\}^2\mid a,x]
\end{align*}
However, this does not imply that all $q_{obs}\in \mathbb{Q}_{obs};$  are in $\mathcal{L}^2_{ax}(Z)$; hence why we specify A.\ref{compUW}(b) for this identification path in the main text rather than relying on these conditions. 

In the alternate path of \citet{cui_semiparametric_2024} (their Remark 5), the completeness condition A.\ref{compZW}(a) is used to imply the uniqueness of $q_{obs}$. Given any two $q_{obs}^1$ and $q_{obs}^2$ in $\mathbb{Q}_{obs},$ they apply it to $E\{q_{obs}^1(Z,a,x)-q_{obs}^2(Z,a,x)\mid W,a,x\}=0$ a.s. to show that $q_{obs}^1(Z,a,x)-q_{obs}^2(Z,a,x)=0$ a.s. This implicitly requires that $q_{obs}^1(Z,a,x)-q_{obs}^2(Z,a,x)$ is square integrable for any two $q_{obs}^1$ and $q_{obs}^2$ in $\mathbb{Q}_{obs},$ thus why we specify $q_{obs}\in\mathcal{L}_{ax}^2(Z)$ for all $q_{obs}\in \mathbb{Q}_{obs}$ in A.\ref{compZW}(b).


In contrast, the identification conditions of \citet{kallus_causal_2022} do not use completeness conditions in identification; thus, such square integrability conditions are not required. However, they do use square-integrability conditions to show properties of their subsequent estimation approaches.

 
\section{The identification condition \eqref{q_id_condition} in a $Z$-conditional form \eqref{bal_prop_idZ}}\label{s:id_z_condtion}
Assuming $\mathbb{Q}_{true}\not=\emptyset$, (A.\ref{PCIq}), consider any pair $(q_{obs},q_{true})\in (\mathbb{Q}_{obs}\times\mathbb{Q}_{true}):$
\begin{align*}
E[&I(A=a)\{q_{obs}(Z,a,X)-q_{true}(Z,a,X)\}Y]\tag{LHS of \ref{q_id_condition}}\\
&=E[E[I(A=a)\{q_{obs}(Z,a,X)-q_{true}(Z,a,X)\}Y\mid Z,A,X]]\tag{by iterated exp.}\\
&=E[I(A=a)\{q_{obs}(Z,a,X)-q_{true}(Z,a,X)\}E[Y\mid Z,A,X]]\tag{pulling functions of $Z,A,X$ out}\\
&=E[I(A=a)\{q_{obs}(Z,a,X)-q_{true}(Z,a,X)\}E[Y\mid Z,A=a,X]].\tag{because $I(A=a)E[Y\mid Z,A,X]=I(A=a)E[Y\mid Z,A=a,X]$}\\
\end{align*}
The fourth line is the LHS of \eqref{bal_prop_idZ}; therefore \eqref{q_id_condition} and \eqref{bal_prop_idZ} are the same condition.

\section{Identification via $q_{obs}$ when $h_{obs}$ is assumed to exist}\label{s:kmu}

Assuming an $h_{obs}$ exists, and letting $\delta_q$ abbreviate $\{q_{obs}(Z,a,X)-q_{true}(Z,a,X)\}$:
\begin{align}
E\big\{&I(A=a)\delta_q E(Y\mid Z,a,X)\big\}\tag{LHS of \ref{bal_prop_idZ}}\\
&=E\big\{I(A=a)\delta_q E(h_{obs}(W,a,X)\mid Z,a,X)\big\}\tag{def of $h_{obs}$}\\
&=E\big\{I(A=a)\delta_q E(h_{obs}(W,a,X)\mid Z,A,X)\big\}\tag{*}\\
&=E\{I(A=a)\delta_q h_{obs}(W,a,X)\big\}\tag{rev. iterated exp.}
\end{align}
The step marked $(*)$ comes from  $I(A=a)E[h_{obs}\mid Z,a,X]=I(A=a)E[h_{obs}\mid Z,A,X]$. Thus, the LHS of \eqref{bal_prop_idZ} is equal to $E[I(A=a)\{q_{obs}(Z,a,X)-q_{true}(Z,a,X)\}h_{obs}(W,a,X)].$

In the main text we explained that if $h_{obs}(W,a,X)$ exists then $q_{obs}$ balances
it by definition, and so does $q_{true}$, because $q_{true}\in \mathbb{Q}_{obs}$. That is, taking $k(W,X)=h_{obs}(W,a,X)$ in \eqref{PCIqw_bal}:
\begin{align*}
E\{h_{obs}(W,a,X)\}&=E\{I(A=a)q_{obs}(Z,a,X)h_{obs}(W,a,X)\}\\
&=E\{I(A=a)q_{true}(Z,a,X)h_{obs}(W,a,X)\},
\end{align*}
so that $E[I(A=a)\{q_{obs}(Z,a,X)-q_{true}(Z,a,X)\}h_{obs}(W,a,X)]=0$.  This step is
what the balancing-weight reading makes transparent; \citet{kallus_causal_2022} obtain it without that motivation.

\section{On identification via outcome bridge function}\label{s:hmirror}

Let $\mathbb{H}_{obs}$ denote the set of ``observed-data'' outcome bridge functions $h_{obs}(W,a,X)$ solving \eqref{PCIh:eq}, and let $\mathbb{H}_{true}$ denote the set of ``true'' outcome bridge functions $h_{true}$ solving \eqref{PCIhu:eq}:
\begin{align}
E(Y\mid Z,a,X)=E\{h_{obs}(W,a,X)\mid Z,a,X\} \quad (a\in\mathcal{A}).\tag{\ref{PCIh:eq}}\\
E(Y\mid U,a,X)=E\{h_{true}(W,a,X)\mid U,a,X\} \quad (a\in\mathcal{A}).\label{PCIhu:eq}
\end{align}

\begin{assumption}\label{PCIhtrue}
    There is a solution to \eqref{PCIhu:eq} (i.e., $\mathbb{H}_{true}\not=\emptyset$).
\end{assumption}
\noindent As in the treatment-bridge case, $\mathbb{H}_{true}\subseteq\mathbb{H}_{obs}$. This can be shown similarly to \S\ref{s:qtrueqobs}:  A.\ref{base}(d) can be used to write $E(Y\mid Z,a,X)=E\{E(Y\mid U,a,X)\mid Z,a,X\}$ and then $W\indep Z\mid U,A,X$ (from A.\ref{base}(e)) to write $E\{E(h\mid U,a,X)\mid Z,a,X\}=E(h\mid Z,a,X)$. Under A.\ref{PCIhtrue}, any $h_{true}$ satisfies $E\{Y(a)\}=E\{h_{true}(W,a,X)\}$; this step uses the $W\indep A\mid U,X$ component of A.\ref{base}(e), via $E\{h_{true}(W,a,X)\mid U,a,X\}=E\{h_{true}(W,a,X)\mid U,X\}$. 
Identification via an observed-data bridge function then reduces to the condition eq. \ref{h_id_condition} and the equivalent forms: 
\begin{align}
    E[h_{obs}(W,a,X)-h_{true}(W,a,X)]&=0\label{h_id_condition}\\
    E\big[E[h_{obs}(W,a,X)-h_{true}(W,a,X)\mid U,a,X]\big]&=0\label{bp_h}\\
    E[I(A=a)\,q_{obs}(Z,a,X)\,\{h_{obs}(W,a,X)-h_{true}(W,a,X)\}]&=0.\label{kmu_h}
\end{align}

The following condition, A.\ref{a_ecp}, plays the role for the outcome bridge that balance propagation (A.\ref{a_balprop}) played for the treatment bridge. 

\begin{assumption}[Error-correction propagation]\label{a_ecp}
For all $a\in\mathcal{A}$ and $P_X$-almost every $x$, $E\{h_{obs}(W,a,x)-h_{true}(W,a,x)\mid U,A=a,X=x\}=0$ almost surely, for every $h_{obs}\in\mathbb{H}_{obs}$ and $h_{true}\in\mathbb{H}_{true}$.
\end{assumption}

\citet{miao_identifying_2018,miao_confounding_2024} provide sufficient completeness conditions for A.\ref{a_ecp}. 
The original result of \citet{miao_identifying_2018} uses completeness of $Z$ for $U$ given $A,X$ (i.e. $E[g(U)\mid Z,a,x]=0\implies g(U)=0$ a.s. for all $a,x$), which implies $\mathbb{H}_{obs}=\mathbb{H}_{true}$ (using parallel logic to \S\ref{s:comp_implies_balprop}) and hence eq.~\ref{bp_h}. Alternatively, error-correction propagation holds trivially when the solution is unique, $\mathbb{H}_{obs}=\mathbb{H}_{true}=\{h_0\}$, so that $h_{obs}-h_{true}=0$; \citet{miao_confounding_2024} give a sufficient condition for uniqueness, namely completeness of $Z$ for $W$ given $A,X$ (i.e. $E[g(W)\mid Z,a,x]=0\implies g(W)=0$ a.s. for all $a,x$). This mirrors A.\ref{compZW}, as used in \citet{cui_semiparametric_2024}. Square integrability conditions (often implicit, as discussed in \S\ref{s:q_true_Lsqr}), will also be needed for these both of these identification paths, in order to invoke completeness conditions. 

Finally, \citet{kallus_causal_2022} (Theorem 1.1) show that when an observed treatment bridge function $q_{obs}$ exists, eq.~\ref{h_id_condition} is equivalent to eq.~\ref{kmu_h}: taking $k(W,X)=h(W,a,X)$ in the balancing condition eq.~\ref{PCIqw_bal} gives $E[h(W,a,X)]=E[I(A=a)\,q_{obs}(Z,a,X)\,h(W,a,X)]$ for any $h$, so from this perspective, the outcome-bridge estimand inherits identification directly from the treatment bridge's balancing property. 

\section{Estimators that assume a linear smoother form of outcome bridge function}\label{s:lin_h}
\subsection{Example 1: Linear outcome bridge function across treatment arms}
While we focus on the linear within-arm $a$ assumption on $h$ in the main text for parallel with the treatment bridge function defined within arms, a more common parametric assumption is that $h(W,A,X)$ is linear \textit{across arms}, as in \citet{tchetgen_tchetgen_introduction_2024,zivich_introducing_2023,cui_semiparametric_2024,liu_regression-based_2024}: 
\begin{equation}\label{h_linear}
    h(W_i, A_i, X_i; \eta) = \eta_0 + \eta_w ^\T W _i+ \eta_a A_i + \eta_x ^\T X_i.
\end{equation}
Combining the moment conditions \eqref{eq:h-moments} across arms, we have:
\begin{equation}\label{eq:h-moments_lin}
    \sum_{i=1}^n[h(W_i,A_i,X_i;\eta)-Y_i]m(Z_i,A_i,X_i)=0_{p}
\end{equation}
where $m(Z_i,A_i,X_i)$ is a $p$-dimensional vector (e.g., $(1,Z_i^\T,A_i,X_i^\T)^\T$). If $p$ is the same as the number of parameters in $\eta$ (and an invertibility condition is satisfied), this can be solved exactly, to derive the implied outcome weights in a similar fashion to those in the per-arm-linear model described in \S\ref{sec_dr_q_equiv} and \S\ref{s:linear_weights}.

\begin{proposition}\label{prop:pgf}
 Let matrix $M$ have rows $m_i^\T=(1,Z_i^\T ,A_i,X_i^\T)$, matrix $R$ have rows $r_i^\T=(1,W_i^\T,A_i,X_i^\T)$ and matrix  $R(a)$ have rows $r_i(a)^\T=(1,W_i^\T,a,X_i^\T).$ Assume a linear form for the outcome bridge $h(W_i, A_i, X_i; \eta) =r_i^\T \eta$, fit using \eqref{eq:h-moments_lin} with $m_i^\T=(1,Z_i^\T,A_i,X_i^\T)$. If $M^\T R$ is invertible, the estimated parameters are: $\hat\eta = (M^\T R)^{-1}M^\T Y$, the fitted values of $h$ for a specific level $a$ are $\hat{h}{(a)} = R{(a)}\hat\eta=S{(a)}Y$ where $S(a) = R{(a)}(M^\T R)^{-1}M^\T $, and the proximal g-computation estimator has an outcome weighted form ($\hat\mu_h(a)=\omega(a)^\T Y$) with implied weights  $\omega(a)^\T  = n^{-1}1_n^\T R{(a)}(M^\T R)^{-1}M^\T =n^{-1}1_n^\T S{(a)}.$
\end{proposition}

\begin{proof}
    In matrix form, equation \eqref{eq:h-moments_lin} can be written as $M^\T (h(a)-Y)=0$. Using the linear form of $h(a)$, this is $M^\T (R\,\eta-Y)=0$. Rearranging, $M^\T R\,\eta=M^\T Y$. Under the invertibility assumption, $\hat\eta=(M^\T R)^{-1}M^\T Y$. The fitted values of the outcome bridge for treatment level $a$ are $\hat h(a)=R(a)\hat\eta=R(a)(M^\T R)^{-1}M^\T Y$ or, defining $S(a)=R(a)(M^\T R)^{-1}M^\T $, $\hat h(a)=S(a)Y$. The proximal g-computation estimator is $\hat\mu_h(a)=n^{-1}1_n^\T\hat h(a)=n^{-1}1_n^\T S{(a)}Y$, an outcome weighted form with $\omega(a)^\T =n^{-1}1_n^\T S{(a)}.$
\end{proof}

The estimator in Proposition \ref{prop:pgf}  also satisfies the conditions of Proposition 1 in the main text, when $\hat q(a)$ is estimated to balance $\mathbb{K}_{wx}$ with rows $(1,W_i^\T,X_i^\T)$. Here, $\hat h(a)$ has the form $\hat h(a) =R(a)\hat\eta$ and $R(a)$ (with rows $(1,W_i^\T,a, X_i^\T)$) is in the span of $\mathbb{K}_{wx}$, as $a$ is a constant. 

\subsection{Example 2: Minimax kernel estimators of \citet{ghassami_minimax_2021}}
 
\citet{ghassami_minimax_2021} propose a nonparametric estimation approach. They begin from a general framework for minimax estimators derived from efficient influence functions for a class of estimands, and then specialize to estimation of both bridge functions in the proximal framework, represented in reproducing kernel Hilbert space function classes. Following their Propositions 3 and 4, fix the arm $a$ and represent each bridge function in a reproducing kernel Hilbert space of functions on $(W,X)$ and on $(Z,X)$, respectively; for simplicity, let the kernels be the same across arms. Let $K_H$ and $K_Q$ be the corresponding kernel Gram matrices, so that:
\begin{align*}
h(w, a,x;\eta,\kappa_h) &= \sum_{j=1}^n \eta_{a,j} \, \kappa_h((w, x), (W_j, X_j)),\\
q(z, a,x;\phi,\kappa_q) &= \sum_{j=1}^n \phi_{a,j} \, \kappa_q((z, x),\; (Z_j, X_j)).
\end{align*}

For example, one could use a Gaussian kernel. Each bridge function is estimated in a minimax framework: $q(Z,a,X)$ is chosen from its class to minimize the worst-case imbalance over functions $h(W,a,X)$ in the other class, and $h(W,a,X)$ likewise against the worst case over $q(Z,a,X)$. This is similar to the estimators of \citet{kallus_causal_2022}, though they more generally draw the adversary from another class of functions, a critic class, and they also accommodate function classes beyond reproducing kernel Hilbert spaces, including linear sieves and neural networks. In both cases each bridge function is fitted against the worst case of the other, and that is what controls the product bias term of the doubly robust functional.

For reproducing kernel Hilbert space classes, \citet{ghassami_minimax_2021} show that for a fixed set of hyperparameters, estimation of the bridge functions yield closed form solutions \citep[Proposition 4 in][]{ghassami_minimax_2021}. This provides a closed form of the implied weights. In summary, for the outcome bridge function: 
\begin{enumerate}
    \item The estimated parameters have the form $\hat\eta_a=(K_{H}D(a)\Gamma_{h,a}D(a)K_{H}+n^{2}\lambda_{h,a} K_{H})^{\dagger}K_{H}D(a)\Gamma_{h,a}D(a)Y$, where $\Gamma_{h,a}=\tfrac14 K_{Q}(n^{-1}K_{Q}+\nu_{h,a}I_n)^{-1}$; here $\lambda_{h,a}$ is the ridge penalty on the bridge function and $\nu_{h,a}$ regularizes the adversarial critic. Here $\Gamma_{h,a}$ is built from $K_Q$: the adversary for $h$ lives on $(Z,X)$-space. Let $L(a)=(K_{H}D(a)\Gamma_{h,a}D(a)K_{H}+n^{2}\lambda_{h,a} K_{H})^{\dagger}K_{H}D(a)\Gamma_{h,a}D(a)$; then $\hat\eta_a$ can be written as $\hat\eta_a=L(a)Y$.
    \item Fitted values for the outcome bridge function are then: $\hat h(a)=K_{H}\hat{\eta}_a=K_{H}L(a)Y=S(a)Y$, with $S(a)=K_HL(a)$.
    \item The estimator $\hat\mu_h(a)=n^{-1}1_n^{\T}K_{H}\hat\eta_a$ then has the weighting representation $\hat\mu_h(a)=\hat\omega_{h,\mathrm{mm}}(a)^{\T}Y$ with
$\hat\omega_{h,\mathrm{mm}}(a)^{\T}=n^{-1}1_n^{\T}S(a)$.
\end{enumerate}

Similarly, the treatment bridge function has the closed form, with: 
\begin{enumerate}
    \item The estimated parameters $\hat\phi_a=(K_{Q}D(a)\Gamma_{q,a}D(a)K_{Q}+n^{2}\lambda_{q,a} K_{Q})^{\dagger}K_{Q}D(a)\Gamma_{q,a}1_n$, where $\Gamma_{q,a}=\tfrac14 K_{H}(n^{-1}K_{H}+\nu_{q,a}I_n)^{-1}$. The penalties $\lambda_{q,a}$ and $\nu_{q,a}$ play the roles that $\lambda_{h,a}$ and $\nu_{h,a}$ played above, and $\Gamma_{q,a}$ is built from $K_H$: the adversary for $q$ lives on $(W,X)$-space.
    \item Fitted values $\hat q(a)=K_{Q}\hat{\phi}_a$,
    \item and outcome weights $\hat\omega_{q,\mathrm{mm}}(a)^{\T}=n^{-1}\hat\phi_a^{\T}K_{Q}D(a).$
\end{enumerate} Combining these via the expression for the proximal AIPW estimator in \S\ref{sec_ow} gives the outcome weights of the minimax doubly robust estimator.

In practice, \citet{ghassami_minimax_2021} propose estimating the bridge functions by cross-fitting \citep{chernozhukov_doubledebiased_2018}. This procedure can still be written in an outcome weighted form; we discuss the outcome bridge first. Consider $j=1,...,J$ folds with size $n_j$. Let $\Pi_j$ be an $n$ by $n$ matrix where a 1 on the $i$th diagonal indicates unit $i$ is in fold $j$ and $D_{-j}(a)=D(a)(I_n- \Pi_j)$---a combination of the matrix that indicates units are in arm $a$, multiplied by a matrix that selects only units not in fold $j$. The estimated $\hat\eta_{a,-j}$  obtained with all units except those in fold $j$ is then $\hat\eta_{a,-j}=L_{-j}(a)Y$, where $L_{-j}(a)=(K_{H}D_{-j}(a)\Gamma_{h,a,-j}D_{-j}(a)K_{H}+(n-n_j)^{2}\lambda_{h,a} K_{H})^{\dagger}K_{H}D_{-j}(a)\Gamma_{h,a,-j}D_{-j}(a)$ and $\Gamma_{h,a,-j}=\tfrac14 K_{Q}((n-n_j)^{-1}K_{Q}+\nu_{h,a}I_{n})^{-1}$. Then:

$$\hat h(a)=\sum_{j=1}^J\Pi_jK_{H}L_{-j}(a)Y, \quad S(a)=\sum_{j=1}^J\Pi_jK_{H}L_{-j}(a), \quad \hat\omega_h(a)=n^{-1}1_n^\T S(a).$$ 

With a similar process on the treatment bridge function (i.e., using $D_{-j}(a)$ and $n-n_j$ for scaling terms instead of $n$), we can obtain $\hat\phi_{a,-j}=(K_{Q}D_{-j}(a)\Gamma_{q,a,-j}D_{-j}(a)K_{Q}+(n-n_j)^{2}\lambda_{q,a} K_{Q})^{\dagger}K_{Q}D_{-j}(a)\Gamma_{q,a,-j}(I_n-\Pi_j)1_n;$ the fitted values $\hat q(a)=\sum_{j=1}^J\Pi_jK_{Q}\hat\phi_{a-j}$; and thus the outcome weights are $\hat\omega_{q,\mathrm{mm}}(a)^{\T}=n^{-1}\hat q(a)^\T D(a).$

\section{Equivalence of estimators when bridge functions are linear in both arms}\label{s:linear_weights}

Proposition \ref{prop:linear_pipw} is proved through the following lemmas. 

\begin{lemma}\label{lemma:linear_pipw}
Let $\mathbb{K}_{zx}$ be the matrix with rows $(1,Z_i^\T,X_i^\T)$ and $\mathbb{K}_{wx}$ the matrix with rows $(1,W_i^\T,X_i^\T)$. Suppose  $q(a)=\mathbb{K}_{zx}\phi_a$ and that $\mathbb{K}_{zx}^{\T}D(a)\mathbb{K}_{wx}$ is invertible. Then the proximal IPW estimator solving \eqref{eq:q-moments} with $k_q^p(W_i,X_i)=(1,W_i^\T,X_i^\T)^{\T}$ is $\hat\mu_q(a)=\hat\omega_{q\text{-lin}}(a)^{\T}Y$ with
$$\hat\omega_{q\text{-lin}}(a)^{\T}=n^{-1}1_n^{\T}\mathbb{K}_{wx}(\mathbb{K}_{zx}^{\T}D(a)\mathbb{K}_{wx})^{-1}\mathbb{K}_{zx}^{\T}D(a).$$
\end{lemma}

\begin{proof}[Proof of Lemma \ref{lemma:linear_pipw}]
Under this set up, \eqref{eq:q-moments} is:
$q(a)^\T D(a) \mathbb{K}_{wx}=1_n^\T \mathbb{K}_{wx}$. Plugging in  $q(a)=\mathbb{K}_{zx}\phi_a$, this is $\phi_a^\T \mathbb{K}_{zx}^\T D(a) \mathbb{K}_{wx}=1_n^\T \mathbb{K}_{wx}$. Under the assumption that $\mathbb{K}_{zx}^{\T}D(a)\mathbb{K}_{wx}$ is invertible, this is $\phi_a^\T=1_n^\T \mathbb{K}_{wx} (\mathbb{K}_{zx}^\T D(a) \mathbb{K}_{wx})^{-1}$. Thus, $\hat q(a)^{\T}=\phi_a^\T \mathbb{K}_{zx}^\T=1_n^\T \mathbb{K}_{wx} (\mathbb{K}_{zx}^\T D(a) \mathbb{K}_{wx})^{-1}\mathbb{K}_{zx}^\T$ and $\hat\mu_q(a)=n^{-1}\hat q(a)^{\T}D(a)Y=\hat\omega_{q\text{-lin}}(a)^{\T}Y$ with $\hat\omega_{q\text{-lin}}(a)^{\T}=n^{-1}1_n^{\T}\mathbb{K}_{wx}(\mathbb{K}_{zx}^{\T}D(a)\mathbb{K}_{wx})^{-1}\mathbb{K}_{zx}^{\T}D(a)$.
\end{proof}

\begin{lemma}\label{lemma:linear_gc}
    As in Lemma \ref{lemma:linear_pipw}, let $\mathbb{K}_{zx}$ be the matrix with rows $(1,Z_i^\T,X_i^\T)$ and $\mathbb{K}_{wx}$ the matrix with rows $(1,W_i^\T,X_i^\T)$. Suppose  $h(a) = \mathbb{K}_{wx}\eta_a$, \eqref{eq:h-moments} is solved with $k_h^p(Z_i,X_i)^{\T}=(1,Z_i^{\T},X_i^{\T})$, and $\mathbb{K}_{zx}^{\T}D(a)\mathbb{K}_{wx}$ is invertible. The corresponding proximal-g-computation estimator can be written in the outcome weighted form $\hat\mu_{lin}(a)=\hat\omega_{lin}(a)^{\T} Y$ with the implied weights 
$\hat\omega_{lin}(a)^{\T}=n^{-1}1_n^{\T}\mathbb{K}_{wx}(\mathbb{K}_{zx}^{\T}D(a)\mathbb{K}_{wx})^{-1}\mathbb{K}_{zx}^{\T}D(a).$

\end{lemma}

\begin{proof}[Proof of Lemma \ref{lemma:linear_gc}]
    In matrix form, equation \eqref{eq:h-moments} can be written as $\mathbb{K}_{zx}^\T D(a) \{h(a)-Y\}=0$. Using the linear form of $h(a)$, this is $\mathbb{K}_{zx}^\T D(a) \{\mathbb{K}_{wx}\eta_a-Y\}=0$. Rearranging, $\mathbb{K}_{zx}^\T D(a)  \mathbb{K}_{wx}\eta_a=\mathbb{K}_{zx}^\T D(a)  Y$. Under the invertibility assumption, $\hat\eta_a=(\mathbb{K}_{zx}^\T D(a) \mathbb{K}_{wx})^{-1}\mathbb{K}_{zx}^\T D(a) Y$. The fitted values of the outcome bridge are $\hat h(a)=\mathbb{K}_{wx}\hat\eta_a=\mathbb{K}_{wx}(\mathbb{K}_{zx}^\T D(a) \mathbb{K}_{wx})^{-1}\mathbb{K}_{zx}^\T D(a) Y$ or, defining $S(a)=\mathbb{K}_{wx}(\mathbb{K}_{zx}^\T D(a) \mathbb{K}_{wx})^{-1}\mathbb{K}_{zx}^\T D(a)$, $\hat h(a)=S(a)Y$. The proximal g-computation estimator is $\hat\mu_h(a)=n^{-1}1_n^\T\hat h(a)=n^{-1}1_n^\T S{(a)}Y$, an outcome weighted form with $\omega(a)^\T =n^{-1}1_n^\T S{(a)}= n^{-1}1_n^\T \mathbb{K}_{wx}(\mathbb{K}_{zx}^\T D(a) \mathbb{K}_{wx})^{-1}\mathbb{K}_{zx}^\T D(a)$
\end{proof}

\begin{proof}[Proof of Proposition \ref{prop:linear_pipw}]
    The conditions of lemmas \ref{lemma:linear_pipw} and \ref{lemma:linear_gc} together are the same as the conditions of \ref{prop:linear_pipw}. The weights derived both are the same: $\hat\omega_{q-lin}(a)^\T=\hat\omega_{h-lin}(a)^\T=n^{-1}1_n^\T S{(a)}= n^{-1}1_n^\T \mathbb{K}_{wx}(\mathbb{K}_{zx}^\T D(a) \mathbb{K}_{wx})^{-1}K_{zx}^\T D(a)$. Finally, the implied weights for the proximal AIPW estimator are also the same. As noted in 
\S \ref{sec_ow}, $\hat\omega_{qh}(a)^\T=\hat \omega_q(a)^{\T}- \hat \omega_q(a)^{\T}\hat S(a) + \hat \omega_h(a)^{\T}$. So, 
$\hat\omega_{qh-lin}(a)^\T=2*\hat \omega_{lin}^{\T}- \hat \omega_{lin}^{\T}\hat S(a)$, and $\hat \omega_{lin}^{\T}\hat S(a)=n^{-1}1_n^\T \mathbb{K}_{wx}(\mathbb{K}_{zx}^\T D(a) \mathbb{K}_{wx})^{-1}\mathbb{K}_{zx}^\T D(a)\mathbb{K}_{wx}(\mathbb{K}_{zx}^\T D(a) \mathbb{K}_{wx})^{-1}\mathbb{K}_{zx}^\T D(a)=n^{-1}1_n^{\T}\mathbb{K}_{wx}(\mathbb{K}_{zx}^{\T}D(a)\mathbb{K}_{wx})^{-1}\mathbb{K}_{zx}^{\T}D(a)=\omega_{lin}^{\T}.$ Therefore, 
$\hat\omega_{qh-lin}(a)^\T=2\hat \omega_{lin}(a)^{\T}- \hat \omega_{lin}(a)^{\T}=\hat \omega_{lin}(a)^{\T}$.
\end{proof}

\begin{remark}
    A caveat for Proposition \ref{prop:linear_pipw} is that the constraint that $q(Z,a,X)$ is linear for all $a\in\mathcal{A}$ is generally incompatible with the constraint that $\sum_{a\in\mathcal{A}} \pr(A=a|W,X)=1$. This is the proximal counterpart of the familiar fact that inverse-probability models cannot be specified linearly in both arms at once. 
\end{remark}


\end{appendix}

\end{document}